\documentclass[conference]{IEEEtran}
\IEEEoverridecommandlockouts
\usepackage{mathrsfs}
\usepackage{amsfonts}
\usepackage{ifpdf}
\usepackage{cite}
\ifCLASSINFOpdf
\usepackage[pdftex]{graphicx}
\else \fi
\usepackage{amsmath}
\usepackage{array}
\usepackage{mdwmath}
\usepackage{mdwtab}
\usepackage{eqparbox}
\usepackage[tight,footnotesize]{subfigure}
\usepackage{algorithmic}
\usepackage{algorithm}
\usepackage{amsmath}
\usepackage{epsfig}
\usepackage{ae}
\usepackage{amssymb}
\usepackage{times}
\usepackage{amsmath}   
\usepackage{booktabs}  
\usepackage{threeparttable} 
\usepackage{multirow}       
\usepackage{xcolor}
\usepackage{amsmath,amssymb,amsfonts}
\usepackage{graphicx}
\usepackage{textcomp}
\usepackage{xcolor}
\usepackage[normalem]{ulem}
\usepackage{amssymb}

\newtheorem{theorem}{\textbf{Theorem}}

\newtheorem{definition}{\textbf{Definition}}

\newtheorem{proposition}{\textbf{Proposition}}

\newtheorem{proof}{Proof}

\def\BibTeX{{\rm B\kern-.05em{\sc i\kern-.025em b}\kern-.08em
		T\kern-.1667em\lower.7ex\hbox{E}\kern-.125emX}}

\newcommand*{\QEDB}{\hfill\ensuremath{\square}}  

\begin{document}
	
\title{Differentially Private Deep Q-Learning for Pattern Privacy Preservation in MEC Offloading}


\author{
	\IEEEauthorblockN{Shuying Gan$^{\dag}$, Marie Siew$^{\P}$, Chao Xu$^{\dag}$$^{\S}$, Tony Q.S. Quek$^{*}$}
	\IEEEauthorblockA{$^{\dag}$School of Information Engineering, Northwest A\&F University, Yangling, Shaanxi, China\\
		$^{\P}$Department of Electrical and Computer Engineering, Carnegie Mellon University, USA\\
		$^*$Information System Technology and Design Pillar, Singapore University of Technology and Design, Singapore\\
		Email: \{ganshuying99, cxu\}@nwafu.edu.cn, msiew@andrew.cmu.edu, tonyquek@sutd.edu.sg
	}
	\thanks{$^{\S}$Corresponding author: Chao Xu, cxu@nwafu.edu.cn}
}

\maketitle
\begin{abstract}
Mobile edge computing (MEC) is a promising paradigm to meet the quality of service (QoS) requirements of latency-sensitive IoT applications. However, attackers may eavesdrop on the offloading decisions to infer the edge server's (ES's) queue information and users' usage patterns, thereby incurring the pattern privacy (PP) issue. Therefore, we propose an offloading strategy which jointly minimizes the latency, ES's energy consumption, and task dropping rate, while preserving PP. Firstly, we formulate the dynamic computation offloading procedure as a Markov decision process (MDP). Next, we develop a \underline{D}ifferential \underline{P}rivacy \underline{D}eep \underline{Q}-learning based \underline{O}ffloading (DP-DQO) algorithm to solve this problem while addressing the PP issue by injecting noise into the generated offloading decisions. This is achieved by modifying the deep Q-network (DQN) with a Function-output Gaussian process mechanism. We provide a theoretical privacy guarantee and a utility guarantee (learning error bound) for the DP-DQO algorithm and finally, conduct simulations to evaluate the performance of our proposed algorithm by comparing it with greedy and DQN-based algorithms.

\end{abstract}
\begin{IEEEkeywords}
mobile edge computing, computation offloading, differential privacy, deep reinforcement learning.
\end{IEEEkeywords}

\vspace{-0.5em}
\section{Introduction}
The Internet of Things (IoT) integrates a large number of pervasive, connected, and smart devices in the physical world via the Internet, enabling various smart IoT applications, e.g., deep-learning-driven smart video surveillance, flying ad hoc networks for precision agriculture, and e-health \cite{7488250}. These IoT applications heavily rely on computationally intensive machine learning algorithms, which are typically resource-hungry and cannot be readily supported by resource-constrained mobile devices (MDs) \cite{8270633}.
Meanwhile, offloading computation to the cloud server (CS) may incur a high latency cost due to the long transmission distance \cite{8030322}.
In light of this, a new computing paradigm called mobile edge computing (MEC) has been proposed, to meet the quality of service (QoS) requirements of latency-sensitive IoT applications. Here, edge servers (ESs) are placed in close proximity to MDs\cite{8736011}, at the network edge (e.g. at base stations or access points). 
For an MEC system, it is of great importance to design efficient computation offloading strategies through achieving appropriate cooperation among ESs and CSs\cite{8664595,9382409,9448034}.
While efficient computation offloading algorithms have been developed for MEC systems under various scenarios \cite{8664595,9382409,9448034}, privacy issues, mainly derived from computation cooperation and data sharing in computation offloading, were ignored. This would potentially provide attackers with the opportunity for privacy mining \cite{9123504}.

Recently, researchers begin to focus on addressing privacy issues in MEC computation offloading, and the available work can be broadly categorized into two types: i.e., data privacy protection \cite{8966451,9461063} and pattern privacy (PP) protection \cite{9669093}. The data privacy issue refers to the case where attackers directly steal users' private information (e.g., account passwords, email addresses, home addresses, etc.) from the transmitted data during computation offloading. Authors in \cite{8966451} designed a sampling perturbation encryption strategy to protect data privacy against attackers during computation offloading. In \cite{9461063}, to protect data privacy, a trustworthy access control mechanism was developed by utilizing smart contracts.
In contrast to data privacy, the PP issue refers to the case where attackers eavesdrop on the offloading patterns and decisions to infer system information, e.g., the size and required computational resources of offloading tasks. Based on this, attackers could further infer users' private information, such as their identities and usage patterns.
To the best of our knowledge, there are very few studies on PP-preserving offloading in MEC systems. In \cite{9669093}, to achieve PP protection, the authors proposed to disguise users' offloading tasks by deliberately generating redundant tasks, sent along with the actual tasks to MEC servers. It achieved a tradeoff between the computation rate and privacy preservation
but inevitably incurred extra computation and communication costs from the generation, transmission, and processing of the redundant tasks.


In this work, a novel PP-preserving dynamic computation offloading strategy is devised for MEC systems.
Particularly, we consider a dynamic computation offloading problem that jointly minimizes the latency, ES's energy consumption, and task dropping rate, and formulates it as a Markov Decision Process (MDP). To solve this problem while addressing the PP issue, we propose a \underline{D}ifferential \underline{P}rivacy \underline{D}eep \underline{Q}-learning based \underline{O}ffloading (DP-DQO) algorithm. The core idea of DP-DQO is to inject noise after the output layer of the deep Q-network (DQN) to make adjacent state-action pairs indistinguishable, where the "distance" between two state-action pairs is evaluated by the difference of their corresponding rewards. This strategy adds noise to the generated offloading decisions, which prevents the attacker from inferring information on the edge queue and users' usage patterns, thereby achieving PP protection. For our proposed DP-DQO algorithm, we provide theoretical analysis on both its privacy guarantee and utility guarantee (learning error bound). Finally, we evaluate DP-DQO's performance through simulations, showing that, by suitably choosing the noise level, DP-DQO can achieve both the high return and PP protection during the computation offloading process.
 
\vspace{-0.5em}

\section{System model and problem formulation} \label{Sec:Section 2}
\subsection{Network Model}
As shown in Fig.1, we consider an MEC system consisting of $N$ mobile devices (MDs) served by an edge server (ES) and a cloud server (CS), where the set of MDs is denoted by $\mathcal N = \{1, 2, \ldots, N\}$. Compared with the CS, the ES is closer to MDs but does not have sufficient computing resources to serve all MDs simultaneously. In this paper, we consider a time-slotted system $\mathcal T = \{1, 2, \ldots, T\}$ where each timeslot's duration is $\Delta_T$ (in seconds). At the beginning of each time slot, the ES may receive offloading tasks from MDs. We consider the task arrivals from MD $n$ following an independent Poisson process with parameter $\lambda _n$ \cite{8457190}. At the beginning of time slot $t$, if there is an offloading task sent from MD $n$, we denote it by $\mathbf x_n(t) \triangleq \left(\rho_n(t), \beta_n(t)\right)$ with $\rho_n(t)$ and $\beta_n(t)$ respectively denoting its data size and required CPU cycles for task computation. For a generic task $\mathbf x_n(t) \triangleq \left(\rho_n(t), \beta_n(t)\right)$, we assume that it is atomic and can not be divided further. Besides, we consider that the data size $\rho_n(t)$ and the required CPU cycles $\beta_n(t)$ follow independent uniform distributions \cite{8664595}.
\begin{figure} [!t]
	\centering
	\leavevmode \epsfxsize=2.8 in \epsfbox{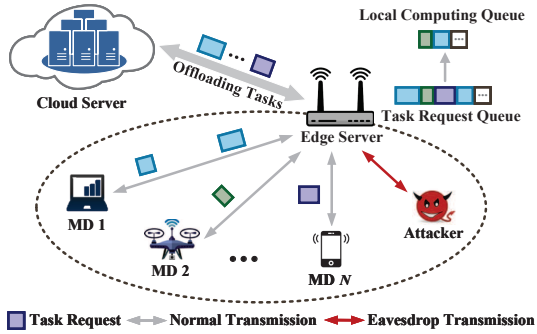}
	\vspace{-1em}
	\centering \caption{Computation offloading model for MEC.} \label{Fig:System_Model}
	\vspace{-1.5em}
\end{figure}

To buffer the arriving offloading tasks, the ES maintains a task request queue (TRQ) whose size is denoted by $q^{TR}_{max}$. The newly arriving tasks would be buffered in the TRQ if there is enough space there, and dropped otherwise. 
After the TRQ is updated in time slot $t$, we denote $I(t)$ as the total number of offloading tasks that have arrived at the ES, and $K(t)$ as the combined size of the tasks buffered at the ES, i.e., $K(t) \leq q^{TR}_{max}$. For the updated TRQ, the ES would deal with the task appearing at its head, denoted by $\mathbf x_H(t)\triangleq \left(\rho_H(t), \beta_H(t)\right)$, deciding between computing it locally (at the ES), or delivering it to the CS via one of $M$ orthogonal channels. At the beginning of time slot $t$, let $\omega(t)$ denote the number of available orthogonal channels, i.e., $\omega(t) \leq M$. Besides, let $A(t)\in\{0, 1\}$ denote the ES's computation offloading decision at time slot $t$, i.e., $A(t)=1$ if the task is offloaded to the CS for processing, and $A(t) = 0$ otherwise. It is worth noting that if there is no free channel between the ES and CS at the beginning of time slot $t$, then the ES has to process the task locally, i.e., $A(t) = 0$ if $\omega(t) =0$.


The ES maintains another queue for the tasks to be processed locally, called the local computing queue (LCQ), according to the first-come-first-serve (FCFS) discipline. We denote LCQ's size by $q^{LC}_{max}$. At the beginning of time slot $t$, we denote the total number of CPU cycles required by the tasks waiting in the LCQ by $\hat{P}(t)$, and the combined size of these tasks by $\hat{K}(t)$ with $\hat K(t) \leq q^{LC}_{max}$. In each time slot, the arriving task $\mathbf x_H(t)\triangleq \left(\rho_H(t), \beta_H(t)\right)$ (if $A(t) = 0$) would be buffered only if there is enough space in the LCQ, i.e., $\hat K(t)+ \rho_H(t)\leq q^{LC}_{max}$. At the end of time slot $t$, we denote by $I_d(t)$ the number of tasks that are dropped due to the overflow of the TRQ or LCQ.

In this paper, we aim at optimizing the computation offloading decisions to jointly minimize the incurred latency, energy consumption of the ES, and task dropping rate. Note that the waiting time of each task in the TRQ is determined by the task arrival processes and is independent of the ES's offloading decisions. 
As such, we equivalently focus on minimizing the latency, ES's energy consumption, and task dropping rate of the tasks appearing at the head of the TRQ, while ignoring their waiting time in the TRQ.
\vspace{-1em}

\subsection{Latency and Energy Consumption Model}
For the task at the head of the TRQ in time slot $t$, i.e., $\mathbf x_H(t)\triangleq \left(\rho_H(t), \beta_H(t)\right)$, the incurred latency, and energy consumption are determined by the offloading decision made by the ES. In this subsection, we specify the latency and energy consumption models regarding local computing at the ES and offloading computation to the CS, respectively.

$\textbf{1) Local Computing at the ES}$: For this case, the latency $L_H^{e}(t)$ consists of two parts, i.e., the waiting time in the LCQ $L_{H}^{w}(t)$ and the computation execution time $L_{H}^{c}(t)$, i.e.,
\vspace{-0.5em}
\begin{align} \label{Eq:Edge_Total_Latency}
	L_H^{e}(t)=L_{H}^{w}(t)+L_{H}^{c}(t) = \left(\hat{P}(t)+\beta_H(t)\right)/{f_{\mathrm{e}}},
\end{align}
in which the waiting time $L_{H}^{w}(t)$ is the sum of the computation execution time of all tasks before the task $\mathbf x_H(t)$ in the LCQ, i.e., $L_{H}^{w}(t) = \hat{P}(t)/{f_{\mathrm{e}}}$ with $f_{\mathrm{e}}$ (in CPU cycles per second) denoting the ES's computation capacity.
Meanwhile, the energy consumption for the local computing at the ES can be expressed as \cite{9093962} 
\vspace{-1em}
\begin{align} \label{Eq:Edge_computing_Energy}
	E_{H}^{e}(t)=\kappa_1\left(f_{\mathrm{e}}\right)^{2} \beta_H(t),
\end{align}
where $\kappa_1$ is an effective capacitance coefficient dependent on the ES's chip architecture, and $\kappa_1\left(f_{\mathrm{e}}\right)^{2}$ denotes the energy consumption per CPU cycle for the local computing.

$\textbf{2) Offloading Computation to the CS}$: For wireless MEC systems, the task transmission time from the ES to the CS is generally much larger compared to the computation time at the CS \cite{9382409}. Therefore, by ignoring the task computation time, we specify the offloading latency $L_{H}^{o}(t)$ as	
\vspace{-0.5em}
\begin{align} \label{Eq:Off_Total_Latency}
	L_{H}^{o}(t)={\rho_H(t)}/{r^{tr}}.
\end{align}
Wherein, $r^{tr}$ denotes the transmission rate from the ES to the CS, which is considered to be constant \cite{8387798}. Besides, the energy consumption of the ES $E_H^{o}(t)$ can be expressed as
 \begin{align} \label{Eq:Off_Total_Energy}
	E_H^{o}(t)=p^{tr}L_{H}^{o}(t),
\end{align}
where $p^{tr}$ denotes the uplink transmission power for the ES.
	
\subsection{Dynamic Computation Offloading Problem Formulation}
Firstly, to balance between the incurred latency and energy consumption, we define the task execution cost $C_0(t)$ of the offloading decision in time slot $t$, i.e, $A(t)$, as
\vspace{-0.3em}	
\begin{equation} \label{Eq:Cost_Function}
	C_0(t)=\left\{\begin{array}{l}
		L_{n}^{e}(t)+\psi E_{n}^{e}(t), \quad \text { if } A(t)=0 \\
		L_{n}^{o}(t)+\psi E_n^{o}(t), \quad \text { if } A(t)=1
	\end{array}\right.
\end{equation}
where $\psi$ is a parameter introduced to balance the trade-off between lowering the latency and reducing the energy consumption. Meanwhile, due to the limited buffer size and channel resources, some tasks may be dropped from TRQ or LCQ. Bearing this in mind, to further reduce the number of dropped tasks, we define the overall cost in time slot $t$ as
\begin{align} \label{Eq:Total_Cost}
	C(t) =C_0(t)/\big[1-I_0(t)\big]=C_0(t)/\big[1-{I_d(t)}/{I(t)}\big]
\end{align}
where $I_0(t)={I_d(t)}/{I(t)}$ denotes the task dropping rate at the end of the time slot. 

In light of this, the dynamic computation offloading problem can be formulated as
\vspace{-0.3em}	
\begin{align} \label{Eq:Utility_Fun}
	\mathbf{P}: & \min _{\mathbf{A}^{T}} \lim_{T \rightarrow \infty} \sum\nolimits_{t=1}^{T} \gamma^{t-1} C(t) \\ \label{Eq:Utility_Obj}
	& \textbf { s.t. } A(t) \in\{0,1\}, \quad \forall t \in \mathcal{T}\\
	& \qquad \omega(t) \leq M, \quad \forall t \in \mathcal{T}
\end{align}
where $\mathbf{A}^{T}=(A(1), A(2),\ldots, A(T))$ denotes the sequence of offloading decisions made by the ES from time slot 1 to $T$, and the discount factor $\gamma \in [0,1)$ is introduced to give importance to the present cost and to make the long-term cumulative cost finite. Besides, constraint (9) indicates that there are at most $M$ orthogonal channels that can be used by the ES in each time slot.

Essentially, problem $\mathbf{P}$ can be formulated as an MDP and solved by reinforcement learning (RL)/ deep reinforcement learning (DRL) algorithms, where the action-value function (or policy) should be learned through trial-and-error interactions with the environment \cite{senadeera2022sympathy}.
In this case, by observing the offloading decisions made by the ES, the attacker (e.g., malicious MDs) can use tools such as inverse reinforcement learning \cite{arora2021survey} 
to infer, for instance, the action-value function. With the inferred action-value function and the observed offloading decision, the attacker can infer the ES's state information \cite{pan2019you}, e.g., the size of buffered tasks and the required CPU cycles. On this basis, the attacker could further obtain the usage patterns of MDs, causing the PP issue. For instance, considering an MEC system consisting of MDs with different usage patterns, the attacker can identify a specific MD from a set of anonymous ones. To address this issue, we propose a novel PP-preserving dynamic computation offloading algorithm, called DP-DQO, so as to minimize the incurred long-term cumulative cost (\ref{Eq:Utility_Fun}) while achieving PP protection during the computation offloading process.


\section{Pattern Privacy (PP) Preserving Solution} \label{Sec:Section 3}

\subsection{MDP Formulation and Differential Privacy}
We formulate the dynamic offloading problem as an MDP consisting of a tuple $(\mathbb S, \mathcal A, R(\cdot,\cdot))$ and depicted as follows:

\textbf{1) State space $\mathbb S$}: The state $\mathcal S(t)$ at the beginning of time slot $t$ is defined as $\mathcal S(t)=({K}(t),\hat{K}(t),\hat{P}(t),\omega(t))$, where ${K}(t)$ denotes the combined size of tasks buffered in the TRQ, $\hat{K}(t)$ and $\hat{P}(t)$ the combined size and required CPU cycles of tasks in the LCQ, and $\omega(t)$ the number of available channels. Denote the space of all possible states by $\mathbb S$.

\textbf{2) Action space $\mathcal A$}: The action of the agent at time slot $t$ is the offloading decision $A(t)$ i.e., $\mathcal A = \{0, 1\}$.

\textbf{3) Reward function $R(\cdot,\cdot)$}: In each time slot, the reward obtained by the agent is dependent on the state $\mathcal S(t)$ and executed action $A(t)$, which is defined as the negative of the overall cost, i.e., $R(\mathcal S(t),{A}(t))= -C(t)$.



In this work, we aim at deriving a policy $\pi^*$ that maximizes the discounted accumulative rewards (i.e., the return) given an initial state $\mathcal{S}(1)$, i.e.,
\vspace{-0.5em}
\begin{align} \label{Eq:Optimal_Policy}
	\nonumber \pi^* 
	\nonumber &= \mathop {\arg }\limits_ \pi \max \mathop {\lim }\limits_{T \to \infty }  \mathbb E \big[ \sum\nolimits_{t=1}^{T} \gamma^{t-1} R(\mathcal S(t),{A}(t)) \left| \mathcal S(1) \right. \big] \\
	& \mathop  =    \mathop {\arg }\limits_ \pi \min \mathop {\lim }\limits_{T \to \infty } \mathbb E \big[\sum\nolimits_{t=1}^{T} \gamma^{t-1}C(t)\big].
\end{align}
To solve this problem while addressing the PP issue, we devise a novel computation offloading algorithm by modifying the DQN-based DRL algorithm with the Function-output Gaussian process mechanism (FGPM).
Particularly, the core idea of our proposed algorithm is to inject the generated noise into the output of the DQN 
to provide 
differential privacy (DP) throughout the offloading process. Before formally presenting our proposed algorithm, we start by introducing some necessary definitions.

\begin{definition} 	\label{Th:differential privacy}
	 ($(\epsilon, \delta)$-DP \cite{dwork2014algorithmic}): A random mechanism $\mathcal{M}$: $\mathcal{R} \rightarrow \mathcal{U}$ with domain $\mathcal{R}$ and range $\mathcal{U}$ satisfies $(\epsilon, \delta)$-DP, if for any two adjacent inputs $r, r^{\prime} \in \mathcal{R}$ and for any subset of outputs $\mathcal{X} \subseteq \mathcal{U}$ we have
	\begin{align} \label{Eq:DP}
		\operatorname{Pr}\left[\mathcal{M}\left(r\right) \in \mathcal{X}\right] \leq e^{\epsilon} \operatorname{Pr}\left[\mathcal{M}\left(r^{\prime}\right) \in \mathcal{X}\right]+\delta
	\end{align}
\end{definition}
where $\epsilon$ denotes the privacy budget and $\delta$ the relaxation factor.
\begin{definition}\label{Th:DP_sen}
	For a random mechanism $\mathcal{M}$: $\mathcal{R} \rightarrow \mathcal{U}$, its sensitivity $\Delta_{\mathcal{M}}$ is defined as the maximum difference between the query results of two adjacent inputs $r, r^{\prime} \in \mathcal{R}$, i.e, \cite{dwork2014algorithmic}
	\begin{align} \label{Eq:DP_sen}
		\Delta_{\mathcal{M}}=\sup _{r, r^{\prime} \in \mathcal{R}}\left\|\mathcal{M}(r)-\mathcal{M}\left(r^{\prime}\right)\right\|
	\end{align}
\end{definition}
where $r, r^{\prime} \in \mathcal{R}$ denotes a pair of adjacent inputs, and $\|\cdot\|$ the an norm function defined on $\mathcal{U}$.

It is noteworthy that for a random mechanism $\mathcal{M}$, the sensitivity is used to quantitatively assess its effect on data privacy. Particularly, the larger sensitivity of the mechanism $\mathcal{M}$, the larger probability of privacy leakage, and the allocated privacy budget $\epsilon$ to make the mechanism $\mathcal{M}$ satisfy DP is also larger.
By considering the action-value function as a random mechanism and the state-action pairs as its inputs, it seems reasonable to develop a DP-based DRL algorithm to solve MDPs while making the learned action-value function satisfy DP. To achieve this, the traditional action-value function defined in RL needs to be modified with the noise mechanism \cite{dwork2014algorithmic}. In this paper, we choose the Function-output Gaussian process mechanism (FGPM) to deal with the continuous inputs, which extends the Gaussian mechanism with the reproducing kernel Hilbert space (RKHS).
\begin{definition} \label{Th:RHKS}
	(Function-output Gaussian process mechanism (FGPM) $F(\cdot)$ \cite{wang2019privacy}): Given a random mechanism $\mathcal{M}$: $\mathcal{R} \rightarrow \mathcal{U}$ with sensitivity $\Delta_{\mathcal{M}}$, the FGPM $F(\cdot)$ is defined as:
	\begin{align} \label{Eq:FG}
		F(r)={\mathcal{M}}(r)+g, \forall r\in \mathcal{R}
	\end{align}
where $g$ denotes the noise sampled from the Gaussian process noise $\mathcal{N}\left(0, \sigma^{2} K\right)$, $\mathcal{U}$ an RKHS with kernel function $K$, and $\sigma$ the noise level that can be specified according to the allocated privacy budget.
\end{definition}

\begin{proposition} \label{Th:G_DP}
	If the privacy budget $\epsilon$ satisfies $0<\epsilon<1$, the relaxation factor $\delta$ and the noise level $\sigma$ satisfy $\sigma \geq \sqrt{2 \ln (1.25 / \delta)} \Delta_{\mathcal{M}} / \epsilon$, then the FGPM $F(\cdot)$ satisfies $(\epsilon, \delta)$-DP.
\end{proposition}

\begin{proof} \label{Pf:Pro}
We refer readers to \cite{wang2019privacy} for the detailed proof.
\end{proof}


In the following two subsections, the details of our proposed algorithm and the theoretical analysis on its privacy and utility Guarantees will be elaborated on respectively.
\vspace{-0.5em}

\subsection{Differentially Private Deep Q-learning based Offloading (DP-DQO) Algorithm}
\vspace{-0.5em}
The details of our proposed DP-DQO algorithm are presented in Algorithm \ref{alg:algorithm1}.
At the beginning of DP-DQO, the experience replay buffer $\mathbb{D}$ is cleared out, the parameters of the Q-network $\theta$ are randomly initialized, and the parameters of the target Q-network are set as $\hat{\theta}=\theta$. When the initialization is completed, the algorithm goes into a loop and the learning process is divided into $\Gamma$ episodes, each of which comprises $T$ time slots. At beginning of each episode, the initial state $\mathcal{S}(1)$, sorted state sets $\{\mathbb S_0\}_{A\in\mathcal A}$, and noise dictionaries $\{\mathcal{G}_A\}_{A\in\mathcal A}$ are respectively initialized. Then, in time slot $t$, an action $A(t)$ is chosen by following the $\mathcal{E}$-greedy policy. After the chosen action is executed, the corresponding experience tuple $(\left(\mathcal{S}(t), A(t), R(\mathcal{S}(t), A(t)),\mathcal{S}(t+1)\right)$ is observed and further stored into the replay buffer $\mathbb{D}$.
\addtolength{\topmargin}{0.01in}

After $\Gamma_{eps}$ episodes are completed, functional noise generating and parameter updating are performed. Specifically, a mini-batch of $\Omega$ experience tuples are randomly sampled from the replay buffer, the set of which is denoted by $\{\left(\mathcal{S}_i, A_i,R_i,\mathcal{S}_{i}^{\prime}\right)| 1\leq i \leq \Omega\}$. Then, the sampled tuples are utilized in sequence to update the noise dictionaries and Q-network. Firstly, for the $i$-th experience tuple, the next state $\mathcal{S}_{i}^{\prime}$ is inserted into the set $\mathbb S_{A_i}$ in the ascending order of the reward. Secondly, a Gaussian process $\mathcal{N}\left(\mu_{A,\mathcal{S}_{i}^{\prime}}, \sigma d_{A,\mathcal{S}_{i}^{\prime}}\right)$ is constructed to generate noise for each action $A$ ($\forall A \in \mathcal{A}$)of state $\mathcal{S}_{i}^{\prime}$, where $\mu_{A,\mathcal{S}_{i}^{\prime}}$ and $\sigma d_{A,\mathcal{S}_{i}^{\prime}}$ respectively denote its mean and variance. 
The expressions of $\mu_{A,\mathcal{S}_{i}^{\prime}}$ and $d_{A,\mathcal{S}_{i}^{\prime}}$ are presented as (14) and (15), where $\Psi=(4 \alpha(z+1) / \Omega)^{-1}$ with $\alpha$, $z$, and $\Omega$ respectively denoting the learning rate, balance factor, and batch size. Besides, $\zeta=\left\|\mathcal{S}^{\prime}_{i}-\mathcal{S}^{\prime-}_{i}\right\|_{2}$, $\nu=\left\|\mathcal{S}^{\prime+}_{i}-\mathcal{S}^{\prime}_{i}\right\|_{2}$ and $\Lambda=\left\|\mathcal{S}^{\prime+}_{i}-\mathcal{S}^{\prime-}_{i}\right\|_{2}$, in which $\mathcal{S}^{\prime+}_{i}$ and $\mathcal{S}^{\prime-}_{i}$ are two adjacent states of $\mathcal{S}^{\prime}_{i}$ in the ordered set $\mathbb S_{A_i}$. And, $G_A(S)$ represents the value of the key $\mathcal S$ in the dictionary $\mathcal{G}_A$.
\begin{align} \label{Eq:Gaussian_p1}
	& \mu_{A,\mathcal{S}_{i}^{\prime}} \!= \!\frac{\left(e^{\Psi \zeta}-e^{-\Psi \zeta}\right) G_A\left(\mathcal{S}^{\prime+}_{i}\right) \!+ \!\left(e^{\Psi \nu}-e^{-\Psi \nu}\right) G_A\left(\mathcal{S}^{\prime-}_{i}\right)}{-e^{-\Psi \Lambda}+e^{\Psi \Lambda}}\\
	& d_{A,\mathcal{S}_{i}^{\prime}}=1-\frac{\left(e^{\Psi \zeta}-e^{-\Psi \zeta}\right) e^{\Psi \zeta}+\left(e^{\Psi \nu}-e^{-\Psi \nu}\right) e^{\Psi \nu}}{-e^{-\Psi \Lambda}+e^{\Psi \Lambda}}
\end{align}

\begin{algorithm}[t!]
	\caption{DP-DQO Algorithm}	
	\label{alg:algorithm1}
	\begin{algorithmic}[1]
		\STATE \textbf{Initialization:} Initialize replay buffer $\mathbb{D}$, Q-network $Q_\theta$, target Q-network $\hat Q_{\hat \theta}$, and training start time as $\Gamma_{eps}$.
		\FOR {$episode = 1, \Gamma$}
		\STATE Set the state $\mathcal{S}(1)=(0, 0, 0, M)$, state sets $\{\mathbb S_0\}_{A\in\mathcal A}=\{\emptyset\}_{A\in\mathcal A}$, and noise dictionaries $\{\mathcal{G}_A\}_{A\in\mathcal A}=\{\emptyset\}_{A\in\mathcal A}$.
		\FOR {$t=1,T$}
		\STATE \textbf{Action selection:} With probability $\mathcal{E}$ select a random action ${A}(t)$, otherwise choose the action ${A}(t)=\arg \max_{{A}\in \mathcal A} Q_{\theta}\left(\mathcal{S}(t), {A}(t)\right)$.
		\STATE \textbf{Acting and observing:} Execute the action $A(t)$, receive reward $R(\mathcal{S}(t), A(t))$, and obtain the new state $\mathcal{S}(t+1)$.
		\STATE \textbf{Refreshing replay buffer:} Store the new transition $(\left(\mathcal{S}(t), A(t),R(\mathcal{S}(t), A(t)),\mathcal{S}(t+1)\right)$ into $\mathbb{D}$.
		\IF{$episode > \Gamma_{eps}$}
		\STATE \textbf{Training:} Sample a mini-batch of transitions $\{\left(\mathcal{S}_i, A_i,R_i,\mathcal{S}_{i}^{\prime}\right)| 1\leq i \leq \Omega\}$ from $\mathbb{D}$. 
		\FOR {$i=1,\Omega$}
		\STATE Insert $\mathcal{S}_{i}^{\prime}$ into the sorted state set $\mathbb S_{A_i}$.
		\FOR {${A} \in \mathcal A$}
		\STATE Build a Gaussian process $\mathcal{N}\left(\mu_{A,\mathcal{S}_{i}^{\prime}}, \sigma d_{A,\mathcal{S}_{i}^{\prime}}\right)$  (14) and (15), sample the noise $g_{A,\mathcal{S}_{i}^{\prime}}$ $\sim\mathcal{N}\left(\mu_{A,\mathcal{S}_{i}^{\prime}}, \sigma d_{A,\mathcal{S}_{i}^{\prime}}\right)$, update the dictionary $\mathcal{G}_A$.
		\ENDFOR
		\STATE Calculate the target $y_i$ with (\ref{Eq:Target}).
		\ENDFOR
		\STATE Update the Q-network $Q_\theta$ with (\ref{Eq:SGD}).
		\STATE \textbf{Update target network:} $\hat{\theta}=\theta$ every $\Gamma_0$ episodes.
		\ENDIF
		\ENDFOR
		\ENDFOR
		\STATE \textbf{Output:} The Q-network $Q_\theta$.
	\end{algorithmic}
\end{algorithm}
For each state-action pair $(\mathcal{S}^{\prime}_{i}, A), \forall A \in \mathcal A$, the noise is generated according to $g_{A,\mathcal{S}_{i}^{\prime}}$, which is further utilized to update the dictionary, i.e., $\mathcal{G}_A = \mathcal{G}_A \bigcup \{\mathcal{S}_{i}^{\prime}, g_{A,\mathcal{S}_{i}^{\prime}}\}$, and
calculate the target $y_i$ with the target Q-network $\hat Q_{\hat \theta}$, i.e.,
\begin{align} \label{Eq:Target}
y_i\! = \! R(\mathcal{S}_i, A_i)+\gamma \arg \max_{{A}^{\prime} \in \mathcal A} \left(\hat{Q}_{\hat \theta}\left(\mathcal{S}_i^{\prime}, {A}^{\prime}\right) + {G}_{A^{\prime}}(\mathcal{S}_i^{\prime})\right).
\end{align}
Then, the average of the noised temporal-difference (TD) errors can be expressed
\begin{align}  \label{Eq:Loss_Fun}
	\mathcal{L} = \!1 / \Omega \sum\nolimits_{i=1}^{\Omega} \big[y_i-\left(Q_{\theta}\left(\mathcal{S}_i, A_i\right)+{G}_{A_i}(\mathcal{S}_i)\right)\big]^{2}
\end{align}
with which the Q-network $Q_{\theta}$ can updated according to
\begin{align} \label{Eq:SGD}
	\theta = \theta-\alpha\nabla_{\theta} \mathcal{L}.
\end{align}
During the training phase, the target Q-network $\hat Q_{\hat \theta}$ is updated by setting $\hat{\theta}=\theta$ every $\Gamma_0$ epochs.

\subsection{Privacy and Utility Guarantees of DP-DQO Algorithm}
In this subsection, we analyze both the privacy guarantee and utility guarantee for our proposed DP-DQO algorithm, as shown in the following two theorems.
\begin{theorem}  \label{Th:Privacy_analysis}
(Privacy Guarantee)
The action-value function learned by DP-DQO is $\left(\epsilon, \delta+ \exp \left(-(2 z-8.68 \sqrt{\Psi} \sigma)^{2} / 2\right)\right)$
-DP, if 
$2 z > 8.68 \sqrt{\Psi} \sigma$ and 
$\nonumber \sigma \geq J(\alpha, z, D, \Omega) \Delta_{F}$
$\sqrt{2((\Gamma-\Gamma_{eps}) T / \Omega) \ln (e+\epsilon / \delta)} / \epsilon$, in which $J(\alpha, z, D, \Omega)$ $= \left((4 \alpha(z+1) / \Omega)^{2}+4 \alpha(z+1) / \Omega\right) D^{2}$, $\Psi=(4 \alpha(z+1)$ $ / \Omega) ^{-1}$, $D$ denotes the Lipschitz constant of the action-value function approximation, $\Omega$ the size of mini-batch, and $(\Gamma-\Gamma_{eps}) T$ the total number of steps in training. Wherein, the two state-action pairs are called adjacent if their "distance", evaluated by the difference of their corresponding rewards (i.e., $R$ and $R^{\prime}$), is less than the sensitivity of the FGPM $\Delta_{F}$, i.e., $\left\|R-R^{\prime}\right\|_1 \leq \Delta_{F}$.
\end{theorem}
\begin{proof}   \label{PR:privacy}
Let $Q$ and $Q^{\prime}$ respectively denote the estimated Q-values given two adjacent state-action pairs, the corresponding rewards of which satisfy $\left\|R-R^{\prime}\right\|_1 \leq \Delta_{F}$. To establish the DP guarantee, we will check the update step in line 17 of Algorithm 1. Firstly, let $\widetilde{Q}$ denote the Q-value before updating the $Q$-network. Then, we have
\begin{align}
	\left\|Q-\widetilde{Q}\right\|_{1} \leq \alpha 	D\left(2+G_A\left(\mathcal{S}_i^{\prime}\right)-G_A\left(\mathcal{S}_i\right)\right) / \Omega.
\end{align}
From Lemma 8 of \cite{wang2019privacy}, $\left\|Q-\widetilde{Q}\right\|_{1} \leq 2\alpha D(z+1) / \Omega$ is satisfied with probability at least $1-\exp \left(-(2 z-8.68 \sqrt{\Psi} \sigma)^{2} / 2\right)$. Similarly, we have $\left\|Q^{\prime}-\widetilde{Q}\right\|_{1} \leq 2\alpha D(z+1) / \Omega$ is satisfied with probability at least $1-\exp \left(-(2 z-8.68 \sqrt{\Psi} \sigma)^{2} / 2\right)$. Furthermore, by the triangle inequality, it can be derived that $\left\|Q-Q^{\prime}\right\|_{1} \leq 4 \alpha D(z+1) / \Omega$ for any $\left\|R-R^{\prime}\right\|_1 \leq \Delta_{F}$. Next, by defining $d=Q-Q^{\prime}$ and resorting to Lemma 6 in \cite{wang2019privacy}, we have
\vspace{-0.5em}
\begin{align} \label{Eq:D_RKHS Norm_Org}
	\|d\|_{\mathcal{H}}^{2} \leq(1+\Psi / 2)(4 \alpha D(z+1) / \Omega)^{2}+D^{2} / 2 \Psi
\end{align}
where $\|\cdot\|_\mathcal{H}$ denotes the RKHS norm. By setting $1/\Psi = 4 \alpha (z+1) / \Omega$, we can transform (\ref{Eq:D_RKHS Norm_Org}) as
\begin{align} \label{Eq:D_RKHS Norm}
\|d\|_{\mathcal{H}}^{2} \leq ((4 \alpha (z+1) / \Omega)^{2}+4 \alpha (z+1) / \Omega)D^{2}.
\end{align}

Referring to Definition 3 and Proposition 1, we inject noise sampled by the Gaussian process to $Q$ that makes the update step achieve $\left(\epsilon, \delta+ \exp \left(-(2 z-8.68 \sqrt{\Psi} \sigma)^{2} / 2\right)\right)$-DP, given that $\sigma \geq \sqrt{2 \ln \left(1.25 / \delta \right)}\|d\|_{\mathcal{H}} / \epsilon$.
This shows that each iteration of update has a privacy guarantee. Finally, according to the composition theorem in \cite{kairouz2015composition}, multiple iterations in our proposed algorithm provide a 
privacy guarantee. \QEDB
\end{proof}

Next, we establish the utility guarantee of DP-DQO by analyzing the discrepancy between the action-value of some state-action pairs learned by our algorithm and that regarding the optimal policy, i.e., the learning error (utility loss) \cite{wang2019privacy}.
In Theorem 2, we provide the upper bound of the learning error for our algorithm, which tends towards zero as the size of the state space approaches infinity, i.e., a utility guarantee is achieved.
\begin{theorem} \label{Th:Convergence}
(Utility Guarantee)
Let $Q$ and $Q^{*}$ respectively denote the action-value of some state-action pairs learned by our algorithm and that regarding the optimal policy, and $n$ the cardinality of the state space $\mathbb S$, i.e., $|\mathbb S|=n$. In the case $n<\infty$, and $\gamma < 1$, the utility loss (learning error) of the algorithm satisfies
\vspace{-0.5em}
\begin{align}
	\mathbb{E}\left[\left\|Q-Q^{*}\right\|_{1}\right] \leq \frac{2 \sqrt{2} \sigma}{\sqrt{n \pi}(1-\gamma)}.
\end{align}
\end{theorem}

\begin{proof} \label{Pf:Convergence}
	Due to space limitations, we omit the derivations and refer readers to \cite{wang2019privacy} for the detailed proof. \QEDB
\end{proof}

\section{Simulation results} \label{Sec:Section 4}
In this section, we conduct simulations to evaluate the performance of our proposed DP-DQO algorithm. Consider an MEC system consisting of a CS, an ES, and $N=5$ MDs. The length of the time slot is set as $\Delta_T=1$. The data size and the required CPU cycles for arriving tasks are uniformly generated between $[5, 50]$ MB and $[0.5 \times 10^{11}, 2.0 \times 10^{11}]$ cycles, respectively. Besides, we respectively set the size of the TRQ and LCQ as $q^{TR}_{max}=5$ GB and $q^{LC}_{max}=2$ GB, and the computation capacity and effective capacitance coefficient of the ES as $f^{\mathrm{e}}=5 \times 10^{10}$ and $\kappa_1={10}^{-11}$ \cite{Zhang2013}.
For the link from the ES to CS, we adopt the channel model and parameters as in \cite{8387798} and correspondingly set the transmission rate as $r^{tr}=5$ MB/s.
The performance of our proposed DP-DQO algorithm is evaluated against two baseline algorithms: 1) the \textit{Greedy algorithm} and 2) the \textit{DQN-based algorithm}. For DRL algorithms, each artificial neural network (ANN) consists of two fully connected hidden layers, each with 128 neurons, where the ReLU function is adopted as the activation function. We set the length of one episode as $T = 100$, the size of replay buffer $\mathbb{D}=2000$, the mini-batch size $\Omega=64$, and the target network update frequency $\Gamma_0=10$. To achieve exploration, we adopt the $\mathcal{E}$-greedy policy with $\mathcal{E}=0.02$. The learning rate $\alpha$ and discount factor $\gamma$ are set to $0.002$ and $0.98$, respectively. Here, all simulation results are obtained by averaging 10 independent runs with different seeds, and for fair comparisons, the same seed is adopted for all algorithms and policies in one run.

\begin{figure} [!t]  \label{Fig:Convergence}
	\centering
	\leavevmode \epsfxsize=2.8in  \epsfbox{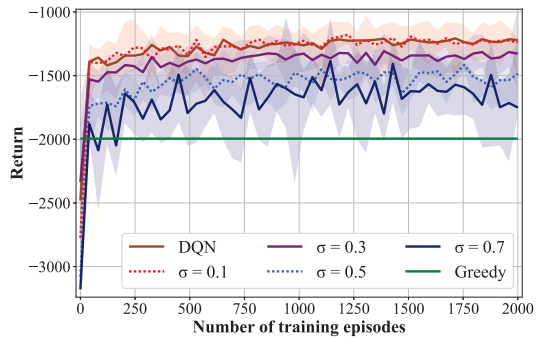}
	\vspace{-0.5em}
	\centering \caption{Convergence comparison for DP-DQO at noise levels $\sigma=\{0.1,0.3,0.5,0.7\}$ and two baseline algorithms.} \label{Fig:Converge_Com}
	\vspace{-0.5em}
\end{figure}

Firstly, we evaluate the convergence of DP-DQO with different noise levels, as shown in Fig. 2. The dark curve (solid or dotted) shows the mean value over runs, and the shaded areas are obtained by filling the interval between the maximum and minimum values over runs. In Fig. 2, it can be seen that the training performance is negatively correlated with the noise level $\sigma$. This is mainly due to the fact that a higher level of privacy protection would reduce the difference between the Q-values of adjacent actions more significantly. In return, this may cause more uncertainty in the Q-network's update direction, thereby making the algorithm's performance more unstable. To further evaluate the effectiveness of DP-DQO, we vary the task arrival rate from $0.1$ to $0.4$, and present the achieved return in Fig. 3. It can be seen that our proposed DP-DQO algorithm outperforms the greedy algorithm in all cases. Besides, it is noteworthy that when the noise level is not high, e.g., $\sigma=0.1$, the return achieved by our proposed DP-DQO algorithm is roughly the same as that of the DQN-based algorithm. In other words, by suitably choosing the noise level, a high utility guarantee (low learning error) can be achieved while preserving the PP protection during the computation offloading process. 

\begin{figure} [!t]
	\centering
	\leavevmode \epsfxsize=2.8in  \epsfbox{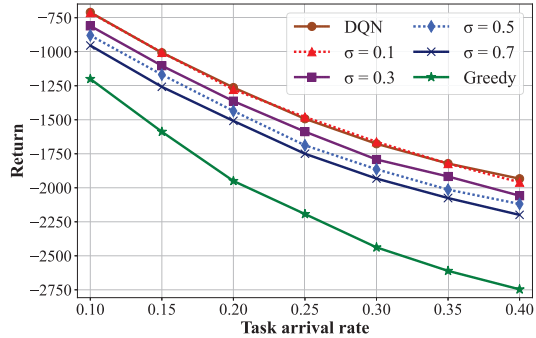}
	\vspace{-0.5em}
	\centering \caption{Performance comparison for DP-DQO at noise levels $\sigma=\{0.1,0.3,0.5,0.7\}$ and two baseline algorithms, in terms of the return where the task arrival rate from 0.1 to 0.4.} \label{Fig:change_N}
	\vspace{-1.5em}
\end{figure}


\vspace{-0.5em}
\section{Conclusions} \label{Sec:Section 5}
In this paper, we have investigated and proposed a computation offloading strategy for the MEC system which not only minimizes the offloading cost, consisting of the latency, energy consumption of the edge server, and task dropping rate but also preserves Pattern Privacy (PP) during the offloading process. This prevents attackers from inferring the edge server's queuing information and users' usage patterns through inverse reinforcement learning and other techniques when they observe the offloading decisions. By modifying the vanilla DQN with the Function-output Gaussian Process Mechanism, a novel PP-preserving dynamic computation offloading algorithm, called DP-DQO, was proposed. For DP-DQO, we provided theoretical analysis on both its privacy guarantee and utility (learning error) guarantee. Our simulation results show that DP-DQO with a suitable noise level performs as well as the DQN-based algorithm in terms of its achieved return and significantly outperforms the greedy algorithm.

\section*{Acknowledgments}
\vspace{-0.25em}
This paper was supported by the National Natural Science Foundation of China (62271413) and the Chinese Universities Scientific Fund (2452017560). 
\vspace{-0.25em}

\bibliographystyle{IEEEtran}
\bibliography{IEEEabrv,Privacy_offloading_Ref}

\end{document}